%% file: main.tex
\documentclass[english]{llncs}
\usepackage{amsmath}
\usepackage{amssymb} 
\usepackage{bbm}
\usepackage{proof}
\usepackage{cmll}
\usepackage{mathbbol}
\usepackage{mathrsfs}
\usepackage{upgreek}
\usepackage{pifont}
\usepackage[all]{xy} 

\usepackage{color}
\usepackage{graphicx}

\renewcommand{\bold}[1]{{\bf #1}}

\newcommand{\Int}[1]{\mathbb{\Lbrack} #1\mathbb{\Rbrack}}

\newcommand{\Th}{\mathrm{Th}}

\renewcommand{\phi}{\varphi}

\newcommand{\Inte}[3]{\mathbb{\Lbrack} #1\mathbb{\Rbrack}_{#2}^{#3}}

\newcommand{\catSet}{\bold{Set}}

\newcommand{\Set}{{\cal S}}
\newcommand{\Cont}{{\cal C}}
\newcommand{\Err}{{\cal E}}
\newcommand{\Lat}{{\cal L}}

\newcommand{\Rel}{{\cal R}}

\newcommand{\SET}{{\Lambda^{\cal S}}}
\newcommand{\CONT}{{\Lambda^{\cal C}}}
\newcommand{\ERR}{{\Lambda^{\cal E}}}

\newcommand{\bool}{{\rm bool}}
\newcommand{\ttt}{t\!t}
\newcommand{\ff}{f\!\!f}
\newcommand{\true}{{\tt{true}}}
\newcommand{\false}{{\tt{false}}}
\newcommand{\ifs}{{\tt{if}}}
\newcommand{\por}{{\tt{por}}}

\newcommand{\arrow}{\rightarrow}

\newcommand{\coll}[3]{{#1}\rightarrow^{#3} {#2}}
\title{
Extensional collapse situations I:
non-termination and
unrecoverable errors}

\author{Antonio Bucciarelli
}

\institute{
PPS, UMR 7126 CNRS and Universit\'e Paris Diderot, France
}

\authorrunning{Bucciarelli}

\begin{document}
\maketitle

\begin{abstract} 
We consider a simple model of higher order, 
functional computations over the booleans. 
Then, we enrich the model in order to encompass
non-termination and unrecoverable errors, taken separately or
jointly. We show that the models so defined  form a lattice
when ordered by the {\em extensional collapse situation} 
relation, introduced in order to compare models 
with respect to the amount of ``intensional information'' that they provide
on computation.
The proofs are carried out
by exhibiting suitable applied $\lambda$-calculi, 
and 
by exploiting the fundamental lemma of logical relations.
\end{abstract}
\section{Introduction}\label{section:introduction}
Properties of programs are often considered  as
split in two categories: the 
 {\em extensional} properties and the {\em intensional} ones. 
Typical intensional properties  
are those concerning  complexity issues, 
while typical extensional ones are those concerning input-output,
observable, behaviours.
This dichotomy originated from recursion theory: 
a predicate ${\cal P}(M)$ on Turing Machines
is extensional if, whenever it holds for a T.M. $M$, it holds
for all the T.M. computing the same function as $M$.
Otherwise,  $\cal P$ is intensional.
So, for instance, whether
the computation on $0$ takes more or less than  $100$ steps
is an intensional property,
while undefinedness on $0$ 
is an extensional one.

It is the intended meaning of a T.M., namely the underlying 
partial recursive function, which determines the frontier
between what is (to be considered as) extensional and what is 
(to be considered as) intensional. 

Therefore, a tentative definition of the notion of extensional 
property could be the following:
a property of programs is extensional if the fact  whether or 
not  a program $P$ enjoys it does affect the denotation of $P$; 
it is intensional otherwise.


Stated this way, being extensional or intensional is a matter
of the intended meaning of programs,  and, as such, it is a relative 
notion\footnote{The very same happens for the dichotomy  ``function vs
algorithm'', which could be misleadingly considered as an absolute one. 
For instance, the well-known parallel-or
``function'' 
is an
algorithm relatively to the model $\Set$ of {\em total objects}, and 
a function  relatively to the model $\Cont$ of {\em partial
objects}, defined below.}.

The following examples, written in a call by name simply typed 
$\lambda$-calculus with constants,  illustrate this claim:
\begin{example}\label{ex:1}
The terms:



(1) {\tt $\lambda$x:bool. true}

(2) {\tt $\lambda$x:bool. if x then true else true  }

get the same interpretation in a model where  
the type of booleans is interpreted by the 
set $\{\ttt,\ff\}$, and higher types by the corresponding 
function space in $\catSet$;
they get different interpretations in a model
that accommodates non-terminating computations,
denoted by a suitable element $\perp$.  

Otherwise stated, strictness is an intensional property relatively to 
the very simple, set-theoretical model, an extensional one relatively to 
a more sophisticated one, based on partial orders.

The linguistic feature that forces strictness to be taken into 
account extensionally is anything allowing for non-termination 
(e.g. fixpoint operators). 
\end{example}
\begin{example}
Similarly, the terms:



(3) {\tt $\lambda$x:bool. $\Omega$}

(4) {\tt $\lambda$x:bool. if x then $\Omega$ else  $\Omega$ }

where $\Omega$ is any diverging term,
get the same interpretation  in  the model of non-termination
introduced above; 
they get different interpretations in a model that accommodates also
errors, denoted by suitable element
$\top$.



\end{example}
More examples could be provided, concerning for instance  the property of 
linearity, 
or the evaluation-order dependency, 
but they go beyond the semantic analysis that we propose in the 
present work.

Hence, the frontier between extensional and intensional properties
is determined essentially by the model we refer to. 
Nevertheless, certain features of programs 
have to be taken into account extensionally by all models, 
namely those features
that affect the operational behaviour of terms (for instance,
strictness in a language allowing for non-termination).

By the way, a fully abstract model is one that  keeps implicit 
(i.e. intensional) all the properties that can be kept implicit.

It appears that the models of simply typed $\lambda$-calculi
may be classified with respect 
to the amount of information on programs that they provide
explicitly, i.e. extensionally.

In this paper, we propose a  way of performing this
classification, and show how it works on some very basic models 
of higher order computation over the booleans.


The basic tool used for the classification is the 
notion of {\em extensional collapse situation}. An
extensional collapse situation is given by two 
models\footnote{The models considered in this paper 
are families 
of sets indexed by simple types, usually called {\em type frames}.}, 
and a binary pre-logical relation between them which is a partial 
surjective function, at all types. In this case, the 
small model, i.e. the target of the surjective function, 
can be considered as what is left of the big one when 
some kind of computational behaviour 
(in this paper, we focus on non-termination and errors) 
is forbidden.

Elements of the big model may be  mapped 
onto the same element of the small one by the surjection: this is the
case, for instance, of the interpretations of the terms (1) and 
(2) above, when passing from the partially ordered model to the 
simple set-theoretic one. They may also ``vanish'', since the surjection
is  partial: 
this is the case of the interpretation of (3) in the aprtially ordered model,
the {\em non-total} constantly undefined function.

In order to prove that a given pair of models $\cal M$, $\cal N$, 
determines 
an extensional collapse situation, we follow the following 
pattern:
\begin{itemize}
\item define a simply typed $\lambda$-calculus $\Lambda^{\cal N}$,
such that $\cal N$ is fully complete  w.r.t. $\Lambda^{\cal N}$
(i.e. such that all elements of  $\cal N$, at all types, are 
definable by closed $\Lambda^{\cal N}$-terms),
\item show that $\cal M$ is a model of  $\Lambda^{\cal N}$,
\item use the fundamental lemma of logical relations 
(Lemma \ref{lemma:fundamental}) to exhibit 
a suitable pre-logical relation, induced by  $\Lambda^{\cal N}$, and conclude.
\end{itemize}

An interesting by-product, which is an immediate 
corollary of the fundamental lemma of logical relations, 
is that the $\Lambda^{\cal N}$-theory of
 $\cal M$ is included in that of  $\cal N$.

We consider the following models:
\begin{itemize}
\item $\Set$: the full type hierarchy over the set 
$\{\ttt,\ff\}$. Its elements are ``higher order boolean circuits''.
\item $\Cont$: the standard model of monotonic\footnote{
In the general case, 
the morphisms of the standard model are the 
Scott continuous functions. Focusing on finite domains, 
the monotonic functions and the Scott continuous ones coincide.} 
functions over the partial order $\perp<\ttt,\ff$.
\item $\Err$: the dual of  $\Cont$, i.e. the model 
of monotonic function over the partial order 
$\ttt,\ff<\top$, modelling unrecoverable errors in the
absence of non-termination.
\item $\Lat$: the model of monotonic functions over the lattice
$\top<\ttt,\ff<\top$, modelling non-termination and errors.
\end{itemize}
We focus on the extensional collapse situations
represented in Figure \ref{ecs}.
\begin{figure}[!t]
\centering
\input{upcf.pstex_t}
\caption{Extensional collapse situations. 
We write $\coll{\cal M}{\cal N}
{\rm E}$ the fact that $\cal M$, $\cal N$ and the 
logical relation $\rm E$ determine an extensional collapse 
situation. 
When $\rm E$ at ground types is a canonical surjection,
clear from the context, we write $\coll{\cal M}{\cal N}{}$.}
\label{ecs}
\end{figure}
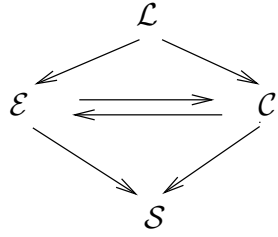
The $\lambda$-calculi used for proving that those are actually
extensional collapse situations are
$\SET$, whose constants are $\true, \false$ and
$\ifs$, and $\CONT$, the parallel extension of finitary PCF, 
due to Plotkin \cite{Plotkin77}. 
It is not necessary to introduce a language $\ERR$,
in order to deal with the  model $\Err$\footnote{
$\ERR$ could be defined as  the dual of  $\CONT$ with respect to the inversion of  $\perp$ and $\top$.}: we prove instead 
that $\Err$ and $\Cont$ are {\em logically isomorphic} 
\footnote{i.e. that they form an extensional collapse situation
in either way.}, and conclude since extensional collapse situations 
do compose.

The existence of an extensional collapse situation 
between two models $\cal M$ and $\cal N$ witnesses the fact that
the target of the collapse is obtained 
by forgetting one (or more) computational aspect(s) that 
are explicitly taken into account in the source.

Let us consider, for instance, the collapse of $\Cont$ over $ \Set$,
described in \cite{BucciarelliS98}: 
unsurprisingly, at the ground type $\bool$ the surjection is the partial 
identity, undefined on $\perp$. 

At first-order types, a monotonic funcion  $f\in\Cont_\sigma$ 
is surjected onto the boolean circuit 
$c\in\Set_\sigma$ exactly when it provides the same value as $c$ on 
total tuples (i.e. on tuples not containing $\perp$). Hence,  
$f$ can be considered as an algorithm implementing $c$.



We have already seen in Example \ref{ex:1} that there exist two implementations 
of the constantly $\ttt$ function from $\Set_\bool$ to $\Set_\bool$, namely the 
strict and the non-strict constantly $\ttt$ functions from $\Cont_\bool$ to $\Cont_\bool$. 

To go a little further, let us consider 
the  $n$-ary  disjunction   $or_n$, of type 
$\underbrace{\bool\arrow\ldots\arrow\bool}_{n\mbox{ times }}\arrow \bool$,
yielding the result $\ttt$ 
whenever at least one of the 
arguments is $\ttt$ and $\ff$ if all the arguments are $\ff$.

They can be implemented in several ways, ranging from the most lazy (and parallel)
algorithm  yielding the result $\ttt$ whenever at least one of the 
arguments is $\ttt$, to the most eager one, yielding the result $\ttt$ whenever all the argument are 
different from $\perp$, and  at least one of them is $\ttt$.

In the case $n=2$, this gives four different implementations of the disjunction, which are 
usually named, from the laziest to the most eager, {\em parallel-or},
{\em left-or},{\em right-or} and {\em strict-or}.

In the general case, it is not difficult to realise that the algorithms 
implementing $or_n$ form a lattice whose size grows exponentially in $n$.
The laziest  algorithm, 
is the bottom, and the most eager
is  the top of the lattice.

Summing up, and generalising, we have that for all types $\sigma$, and all 
boolean functionals $f\in \Set_\sigma$, the set of implementations of $f$ 
is a sub cpo of $\Cont_\sigma$, called the {\em totality class} of $f$, which is a lattice. 

Hence, the model $\Set$ is obtained by collapsing $\Cont$ via the totality relation;
in the same way $\Cont$ is obtained by collapsing $\Lat$ via the ``error-freedom''
relation, which, at ground type, is the partial identity function, 
undefined on $\top$.
 
\subsection{Related works}

In the literature, a model of the simply typed $\lambda$-calculus
is called extensional if all
its elements, at all types, are invariant with respect to
the logical relation defined as the identity at ground 
types\footnote{Otherwise stated: $\forall \sigma, \tau,
f,g:\sigma\arrow\tau\ \mbox{ if } \forall x\in\sigma 
\ f(x)=g(x)\mbox{ then } f=g$.}.

When a model is {\em not} extensional, its extensional 
collapse is performed by eliminating the non-invariant elements. 
The result is an extensional model.

This pattern has been followed , for instance, 
for game models \cite{AbramskyJM00,HylandO00}, or models 
obtained by {\em sequentiality relations} \cite{Sieber92,OHearnR95}. 

Sometimes, the resulting extensional model happens to have been
defined and studied independently: those are  instances
of what we call here extensional collapse situation. 

Examples\footnote{It has to be noticed that in these
examples the collapse is not a logical relation. 
It is defined  as a categorical notion,
independently from the hierarchy of simple types and obviously
instantiating to an extensional collapse situation in the sense
defined in this paper.} of this kind are:
\begin{itemize}
\item   sequential algorithms collapsing on 
strongly stable functions \cite{Ehrhard96a}.
\item the relational model collapsing on 
Scott-continuous functions (between  complete lattices)  \cite{Ehrhard09}.
\end{itemize}
Nevertheless, extensional collapse situations as defined in the
present paper cover a broader 
landscape than these extensional collapses. The essential difference
is  that the extensional collapse of a model is, by definition, unique,
whereas different extensional collapse situations may concern a given
model (as it is the case of the lattice model $\Lat$, 
which collapses over $\Cont$, $\Err$, and $\Set$). 
Even choosing two particular models $\cal M$ and $\cal N$, 
there can be more than one extensional collapse situation between 
them. For instance,  an  extensional collapse situation between 
$\Set$ and $\Cont$
has already been showed  in \cite{BucciarelliS98}, using 
the {\em totality} logical relation. 
The use of the language $\SET$ and of the associated 
(canonical) pre-logical relation provides a different 
partial surjection between those two models, and, by the way,  
makes the proof presented here easier.

Extensional collapse situations, defined in a slightly different
way, have been used in \cite{Bucciarelli97c} for constructing,  
given two models $\cal M$, $\cal N$ of a given applied
$\lambda$-calculus, a quotient model ${\cal M}/{\cal N}$ whose theory 
is a super-set of both ${\rm Th}(\cal M)$ and  ${\rm Th}(\cal N)$.




\section{Logical versus pre-logical relations}

Logical relations have been introduced by Plotkin in 
\cite{Plotkin80}, in order to  characterize
$\lambda$-definability in the full type hierarchy.
The main features of logical relations are:
\begin{itemize}
\item they contain all the tuples of the form 
$(\Int{M}{}{},\ldots,\Int{M}{}{})$, where $M$ is a 
closed, simply typed $\lambda$-term. This is known as 
``the fundamental property''of logical relations.
\item they are completely determined  by their behaviour at ground types.
Hence, defining a logical relation boils down to specify 
it at ground types.
\end{itemize}
The main drawback of logical relations, from our point of view at least, 
is that they are  not closed under (relational) composition;
actually, the property stated in the second item above is
incompatible with the closure under composition. 
In \cite{HonsellS99}, Honsell and Sannella
define {\em pre-logical relations}, as a weakening of the notion of
logical relation which enjoy the fundamental property and is 
closed under (relational) composition.
The price to pay is that pre-logical relation are not determined
by their behaviour at ground types, and that they are 
{\em language-dependent}: at higher-types, pre-logical relations 
have to obey a closure condition which depends from the  specific 
simply typed $\lambda$-calculus $\Lambda$ one is interested in
\footnote{Moreover, just like logical relations, they have to 
be closed with respect to the interpretations of the constants of $\Lambda$.}.

Nevertheless, pre-logical relations are very well adapted to
our framework: in general, we will prove the existence of 
an extensional collapse situation  $\coll{\cal M}{\cal N}
{\rm E}$ by exhibiting a $\lambda$-calculus $\Lambda^{\cal N}$,
with respect to which $\cal N$ is fully complete. In such 
a situation, if $\cal M$ is a model of $\Lambda^{\cal N}$,  the relation 
$\{(\Int{P}^{\cal M},\Int{P}^{\cal N})\ |\ P\in (\Lambda^{\cal N}_\sigma)^0\}_{\sigma\in  Types}$
is a partial surjective function at all types, and it is a 
$\Lambda^{\cal N}$-pre-logical
relation (which is not logical in general). 

We want to be able to compose such pre-logical relations.
Of course, the composition of relations which are partial surjective 
functions is a partial surjective function. 
Moreover, we know that, given a simply typed $\lambda$-calculus $\Lambda$,
the composition of two $\Lambda$-pre-logical relations is
$\Lambda$-pre logical \cite{HonsellS99}.

The point here is that we will compose relations which are pre-logical
relatively to different $\lambda$-calculi: the set of constants of 
 $\Lambda^{\cal N}$ will of course depend on $\cal N$.
Nevertheless, when composing a  $\Lambda^{\cal N}$-pre-logical relation and
a $\Lambda^{\cal N'}$-pre-logical relation, one gets a $\Lambda$-pre-logical
relation for any language $\Lambda$ which is a sub-language of
both $\Lambda^{\cal N}$ and $\Lambda^{\cal N'}$\footnote{This is an immediate consequence of the definition of pre-logical relations and of proposition
\ref{prop:comp}.}. Hence,  we consider  the 
poorest possible language, namely the simply typed $\lambda$-calculus
{\em without} constants, over a given set of ground types (for our purposes, 
a unique ground type $\bool$ is enough), call it $\Lambda$,
and define extensional collapse situations as  $\Lambda$-pre-logical
relations which are partial surjctive functions at all types.

It has to be noticed that logical relations are, {\em a fortiori}, 
$\Lambda$-pre-logical. Hence the composition of two logical 
relations, which is not logical in general, is  $\Lambda$-pre-logical. 

So, in general, we can exhibit different extensional collapse 
situations between two given models. 
Let us consider again, for instance,  the case
of $\cal C$ and $\cal S$. We will see in Section \ref{main} that 
the simply typed $\lambda$ calculus equipped with the ground boolean constants
$\true$ and $\false$ 
and the first-order constant $\ifs$, noted $\Lambda^{\cal S}$, determines,
following the lines described above,  an 
extensional collapse situation  $\coll{\cal C}{\cal S}{E}$,
which is different from the totality collapse described in 
Section \ref{section:introduction}. For instance, the
parallel-or function is related to the binary disjunction in the
totality collapse, but not in $E$, since no term of  $\Lambda^{\cal S}$
defines the parallel-or (by the way, the three ``sequential'' disjunctions
strict-or, left-or and right-or are $\Lambda^{\cal S}$-definable).


\section{Type frames, logical and pre-logical  relations}\label{sec:type frames}
Since all the models considered in this paper are extensional,  we
provide here a very simple  definition of type frame, where higher
types are interpreted by sets of {\em functions}\footnote{In 
section \ref{section:conclusion} we provide  the more general 
definition, encompassing non extensional models.}. 

The set of {\em simple types} over a set $K$ of {\em ground type constants}
is the smallest set containing $K$ and closed w.r.t. the operation
$\sigma,\tau\mapsto \sigma\arrow\tau$.
\begin{definition}
A {\em type frame} $\cal M$ over a set $K$ of ground types is a family of sets
indexed by simple types over $K$, such that 

${\cal M}_{\sigma\arrow\tau} \subseteq \{f\ |\ f \mbox{ is a function from }
{\cal M}_{\sigma} \mbox{ to } {\cal M}_{\tau}\}$.

A type frame is {\em finite} if all the sets of the family are finite sets.
\end{definition}

When ${\cal M}_{\sigma\arrow\tau} = {\cal M}_{\tau}^{{\cal M}_{\sigma}}$,
the type frame is {\em full}. 
\begin{definition}
Given two type frames  $\cal M$ and  $\cal N$ over $K$, a {\em binary
logical relation} $\rm R$ between $\cal M$ and  $\cal N$ is 
a family of binary relations ${\rm R}_\sigma\subseteq {\cal M}_\sigma\times
{\cal N}_\sigma$, indexed by simple types over $K$, such that,
for all $\sigma,\tau, f\in  {\cal M}_{\sigma\arrow\tau}, g\in   
{\cal N}_{\sigma\arrow\tau}$:

$(f,g)\in {\rm R}_{\sigma\arrow \tau}\Leftrightarrow 
[\forall x\in  {\cal M}_{\sigma}, y \in  {\cal N}_{\sigma}\ 
[(x,y)\in    {\rm R}_{\sigma}\Rightarrow (f(x),g(y))\in  {\rm R}_{\tau}]]$
\end{definition}

For defining pre-logical relation, we use the notion of environment-based
interpretation of the  typed $\lambda$-calculus in type frames,
which will be introduced in Section \ref{section:syntax}.
In the following definition, let $\Lambda$ be the simply typed $\lambda$-calculus, without constants,
over a  set $K$ of ground types.

\begin{definition}
A {\em binary
pre-logical relation} $\rm R$ between $\cal M$ and  $\cal N$ is 
a family of binary relations ${\rm R}_\sigma\subseteq {\cal M}_\sigma\times
{\cal N}_\sigma$, indexed by simple types over $K$, such that:
\begin{itemize}
 \item
for all $\sigma,\tau, f\in  {\cal M}_{\sigma\arrow\tau}, g\in   
{\cal N}_{\sigma\arrow\tau}$:
$(f,g)\in {\rm R}_{\sigma\arrow \tau}\Rightarrow 
[\forall x\in  {\cal M}_{\sigma}, y \in  {\cal N}_{\sigma}\ 
[(x,y)\in    {\rm R}_{\sigma}\Rightarrow (f(x),g(y))\in  {\rm R}_{\tau}]]$
\item   
Given:
\begin{itemize}
\item $P\in\Lambda_\tau$, for any type $\tau$,
\item a $\cal M$-environment $\rho$ and a $\cal N$-environment
$\rho'$ such that, for all types $\delta$ and variable $x:\delta$, 
$(\rho(x),\rho'(x))\in {\rm R}_\delta$,
\item a variable $z:\sigma$, for any type $\sigma$,
\end{itemize}

the following holds:

$[\forall (a,b)\in    {\rm R}_{\sigma}\ (\Int{P}^{\cal M}_{\rho[z:=a]},
\Int{P}^{\cal N}_{\rho'[z:=b]}) \in {\rm R}_\tau]\Rightarrow
(\Int{\lambda z.M}^{\cal M}_\rho,\Int{\lambda z.M}^{\cal N}_{\rho'})\in 
 {\rm R}_{\sigma\arrow\tau}$
\end{itemize}
\end{definition}
It is easy to see that any logical relation is pre-logical.
We are interested in a particular kind of pre-logical relations:
\begin{definition}
A binary pre-logical relation $\rm R$ between two type frames $\cal M$ and $\cal N$ 
over $K$ is a {\em pre-logical surjection} if:
\begin{itemize}
\item  at all types, $\rm R$ is surjective:
$\forall\sigma\ \forall y\in  {\cal N}_{\sigma}\ \exists  x\in  {\cal M}_{\sigma}\ 
(x,y)\in \rm R_\sigma$
\item  at all types, $\rm R$ is a partial function:
$\forall\sigma\ \forall x\in  {\cal M}_{\sigma}\ \forall  y,y'\in  {\cal N}_{\sigma}\ 
(x,y),(x,y')\in {\rm R}_\sigma\Rightarrow y=y'$
\end{itemize}
\end{definition}
\begin{lemma}\label{lemma:reduction}
A binary pre-logical relation between $\cal M$ and $\cal N$  
which is a partial function at ground types and
surjective at all types is a pre-logical surjection.
\end{lemma}
\begin{proof}
Straightforward, by induction on types. The extensionality of $\cal N$
is required here.
\end{proof}
\begin{definition}
An {\em extensional collapse situation} 
is a triple ${\cal M},{\cal N}, {\rm  E}$,
where  ${\cal M}$ and ${\cal N}$ are type frames over a given set $K$ of ground types
and  ${\rm E}$ is a pre-logical surjection from ${\cal M}$ to ${\cal N}$.
\end{definition}
We note $\coll {\cal M}{\cal N}{\rm E}$ the fact that  ${\cal M},{\cal N}, {\rm  E}$
is an extensional collapse situation.

Extensional collapse situations are closed under composition, in the following sense:
\begin{proposition}[\cite{HonsellS99}, Prop. 5.5] \label{prop:comp}
If  $\coll {\cal M}{\cal N} {\rm E}$ and  $\coll {\cal N} {\cal P} {\rm F}$
then  $\coll {\cal M} {\cal P}  {{\rm F}\circ{\rm E}}$, where
$({\rm F}\circ{\rm E})_{\sigma}={\rm F}_{\sigma}\circ {\rm E}_{\sigma}$.
\end{proposition}
Hence, the set of type frames over a given set $K$ of ground types
is pre-ordered by the relation ${\cal N}\leq{\cal M} \Leftrightarrow
\coll  {\cal M}{\cal N} {\rm E}$, for some  ${\rm E}$.
Let  $\equiv$ denote the equivalence relation associated to this 
pre-order.

We call {\em pre-logical isomorphism}
a pre-logical relation ${\rm E}$ such that both  ${\rm E}$ and ${\rm E}^{-1}$
are pre-logical surjections.

\subsection{Type frames over {\bool}}
In order to define the type frames we are interested in, 
let us introduce the following notation: if $A$ and $B$ are sets (resp. partially ordered sets),
$A\Rightarrow B$ (resp. $A\Rightarrow_m  B$)
denotes the set of all the functions from $A$ to $B$
(resp. the partially ordered set of monotonic functions from $A$ to $B$, 
ordered pointwise).

We consider four  type frames over a single ground type 
$\bool$: $\Set$, $\Cont$, $\Err$ and $\Lat$, defined as follows:
\begin{itemize}
\item 

\begin{itemize}
\item 
$\Set^{\bool}=\{\ttt,\ff\}$,  
\item $\Set^{\sigma\arrow\tau}=\Set^\sigma \Rightarrow\Set^\tau$.
\end{itemize}

\item 
\begin{itemize}
\item 
$\Cont^{\bool}=\{\perp,\ttt,\ff\}$, partially ordered by 
$\perp < \ttt,\ff$, 
\item $\Cont^{\sigma\arrow\tau}=\Cont^\sigma \Rightarrow_m\Cont^\tau$.
\end{itemize}

\item 
\begin{itemize}
\item
 $\Err^{\bool}=\{\ttt,\ff,\top\}$, partially ordered by 
$ \ttt,\ff< \top$,
\item
$\Err^{\sigma\arrow\tau}=\Err^\sigma \Rightarrow_m\Err^\tau$.
\end{itemize}

\item 
\begin{itemize}
\item
$\Lat^{\bool}=\{\perp,\ttt,\ff,\top\}$, partially ordered by 
$\perp < \ttt,\ff< \top$, 
\item
 $\Lat^{\sigma\arrow\tau}= \Lat^\sigma \Rightarrow_m\Lat^\tau$.
\end{itemize}

\end{itemize}

The type frames $\Cont$ and $\Err$ are dual. We show 
that $\Cont\equiv\Err$ by exhibiting a pre-logical isomorphism.
\begin{proposition}\label{prop:errcont}
The logical relation $\rm E$ between  $\Cont$ and  $\Err$ defined by
 ${\rm E}_\bool=\{(\ttt,\ttt),(\ff,\ff),(\perp,\top)\}$ is a 
pre-logical isomorphism.
\end{proposition}
\begin{proofsketch}
First of all, the relation defined above is logical, hence, {\em a fortiori},
pre-logical. 
In order to show that it is a pre-logical isomorphism, one can prove
by simultaneous induction on types, the following statements:
\begin{enumerate}
\item[(i$_\sigma$)]
 $\rm E_\sigma$ is a bijection.
\item [(ii$_\sigma$)]
$\forall (x,y),(x',y')\in E_\sigma\  (x\leq x' \Leftrightarrow y'\leq y)$.
\end{enumerate}
\end{proofsketch}

\section{Applied $\lambda$-calculi}\label{section:syntax}

The type frames introduced in the previous section are models 
of simply typed $\lambda$-calculi endowed with 
constants, called applied $\lambda$-calculi:

\begin{definition}
Given a set $K$ of ground types, an {\em applied $\lambda$-calculus}
$\Lambda$ over $K$ is given by a family of typed constants 
$C(\Lambda)_\sigma$, indexed
by simple types over $K$. 

The terms of the calculus are simply typed 
$\lambda$-terms built by application and $\lambda$-abstraction 
starting from the typed constants and variables.

The operational semantics of an applied $\lambda$-calculus is specified
by a set of $\delta$-rules, stipulating the behaviour of the constants,
and by the  $\beta$-rule.
\end{definition}
As a matter of notation, we will write $\Lambda_\sigma$ 
(resp. $\Lambda^0_\sigma$) for 
the 
set of terms (resp. closed terms) 
of type $\sigma$.

The usual, environment based, interpretation of an applied $\lambda$-calculus in a type frame $\cal M$,
given a suitable map $C$ from the constants of the 
calculus to elements of the 
appropriate type, is defined as follows:
\begin{itemize}
\item $\Inte{x}{\rho}{\cal M}=\rho(x)$, $\Inte{c}{\rho}{\cal M}=C(c)$

\item $\Inte{PQ}{\rho}{\cal M}=\Inte{P}{\rho}{\cal M}( \Inte{Q}{\rho}{\cal M})$, 
$\Inte{\lambda x:\sigma.  P}{\rho}{\cal M} = d\in {\cal M}_\sigma\mapsto
\Inte{ P}{\rho[x\leftarrow d]}{\cal M}$
\end{itemize}
The interpretation of a closed term $P$  will be noted simply
$\Inte P {} {\cal M}$, or $\Inte P {} {}$, when non ambiguous. 

Two conditions have to be satisfied for a type frame $\cal M$ to be
a model of a given applied $\lambda$-calculus:

\begin{enumerate}

\item In the fourth item of the definition above, one has that 
the function 
$d\in {\cal M}_\sigma\mapsto
\Inte{ P}{\rho[x\leftarrow d]}{\cal M}$
is an element of ${\cal M}_{\sigma\arrow\tau}$,
for the appropriate type $\tau$.
\item The map $C$ is sound:  for any
$\delta$-rule $P\arrow^\delta Q$, one has $\Inte P {} {\cal M}=
\Inte Q {} {\cal M}$.
\end{enumerate}

When these conditions are satisfied, the model is such that 
$P\arrow^* Q\Rightarrow  \Inte P {} {\cal M}=
\Inte Q {} {\cal M}$, where 
the rewriting relation $\arrow$ is the contextual closure of 
$\arrow^\beta\cup\arrow^\delta$.

Note that, for the type frames $\Set$,$\Cont$,$\Err$ and $\Lat$,
the first condition is always satisfied, since for  all of them
 $\Int{\sigma\arrow\tau}$ is the exponential object
$\Int{\tau}^{\Int{\sigma}}$ in some Cartesian closed category
(actually there exists a ccc which is an ambient category
for all those type frames, namely the category of
finite partial orders and monotone functions).
Hence, in order to show that a given type frame 
among $\Set$,$\Cont$,$\Err$ and $\Lat$ is a model 
of a given applied $\lambda$-calculus, it is sufficient
to provide a sound interpretation of the constants of the language.


Among the models of a given applied $\lambda$-calculus, the
fully complete ones are those whose elements are definable 
(in general, one asks for the definability of {\em finite}
elements. We skip this condition here since we focus on finite type
frames).

\begin{definition}
A model $\cal M$ of an applied $\lambda$-calculus $\Lambda$
is {\em fully complete} if for all $\sigma$ and for all
$d\in{\cal M}_\sigma$ there exists a closed $\Lambda$-term $D$ of type 
$\sigma$ such that $\Inte D {} {\cal M}=d$.
\end{definition}

We define now the  applied $\lambda$-calculi that will be used in the sequel 
to prove some extensional collapse situations. 

We call these calculi $\SET$ and  $\CONT$ respectively, 
to emphasise the full completeness of the corresponding type frames.
\subsection{The basic boolean calculus}
Here is the definition of the constants of $\SET$  and of their operational semantics:
\begin{definition}
\begin{itemize}
\item
$C(\SET)_{\bool}=\{\true,\false\}$
\item
$C(\SET)_{\bool\arrow\bool\arrow\bool\arrow\bool}=\{\ifs\}$
\item
$\ifs\ \true\ M\ N\arrow^\delta M$,
$\ifs\ \false\ M\ N\arrow^\delta N$
\end{itemize}
\end{definition}
\begin{lemma}\label{lemma:set}
$\Set$ is a fully complete model of  $\SET$.
\end{lemma}
\begin{proofsketch}
First of all, the interpretation  mapping the ground constants 
on the corresponding boolean values and such that 
 $\Inte{\ifs}{}{\Set} d\ e\ f=\left\{
\begin{array}{ll}
e&\mbox{ if } d=\ttt\\
f&\mbox{ if } d=\ff\\
\end{array}\right.
$ 
is clearly sound, hence $\Set$ is a model of $\SET$.

Concerning full completeness : consider an element 
$f=\{(a_1,b_1),\ldots,(a_n,b_n)\}\in \Set_{\sigma\arrow\tau}$.
If the $a_i$, the $b_j$ and the equality predicate on 
$\Set_\sigma$ were $\SET$-definable, then it would be easy to
write down a $\SET$-term defining $f$, as  a sequence
of nested $\ifs$. 

On the other hand, one can $\SET$-define the equality predicate on
$\Set_{\sigma\arrow\tau}$  if one is able to define
the elements of $\Set_\sigma$ and the equality predicate on 
 $\Set_\tau$.

Hence, the proof consists in a straightforward simultaneous 
induction on types, for the two following properties:

DEF($\sigma$): all the elements of $\Set_\sigma$ are 
$\SET$-definable.

EQ($\sigma$): the equality predicate on $\Set_\sigma$ is
$\SET$-definable.

By the way, EQ($\bool$) is provided by the term 

{\tt $\lambda$x:bool y:bool. $\ifs$ x ($\ifs$ y $\true$ $\false$)
($\ifs$ y $\false$ $\true$)}.



Let $P(\sigma)= \forall e\in \Set_\sigma\exists P\in\SET_\sigma^0\ \Int{P}\Set=e$,


\end{proofsketch}
\begin{lemma}\label{lemma:conterrset}
$\Cont$ and $\Err$ are models of $\SET$.
\end{lemma}
\begin{proof}
The ground constants $\true$ and $\false$ are interpreted by 
the corresponding booleans, in both models. Here is the 
interpretation of the constant $\ifs$ in the two models:

\begin{tabular}{ccccccc}
$\Inte{\ifs}{}{\Cont} d\ e\ f=\left\{
\begin{array}{ll}
e&\mbox{ if } d=\ttt\\
f&\mbox{ if } d=\ff\\
\perp&\mbox{ otherwise}\\
\end{array}\right.
$ &&&&&

$\Inte{\ifs}{}{\Err} d\ e\ f=\left\{\begin{array}{ll}
e&\mbox{ if } d=\ttt\\
f&\mbox{ if } d=\ff\\
\top&\mbox{ otherwise}\\
\end{array}\right.
$ 
\end{tabular}

these  interpretations  clearly validate the $\delta$-rules concerning 
$\ifs$.
\end{proof}
\subsection{Adding non-termination}
Non-termination can be added to the basic calculus by means
of fixpoint combinators, for instance. Nevertheless, a single new 
constant $\Omega:\bool$, whose intended interpretation is 
the undefined boolean value $\perp$ is enough. In fact, 
fixpoint combinators do not add expressive power in the 
case of finite ground types, apart from the possibility
of divergence. In order to achieve the full completeness of 
$\Cont$, we add a parellel-or constant to the language, 
following  Plotkin \cite{Plotkin77}.

$\CONT$ is the  extension of $\SET$
defined as follows:

\begin{definition}
$\forall \sigma \ C(\SET)_{\sigma}\subseteq C(\CONT)_{\sigma}$ and:
\begin{itemize}
\item 
 $\Omega\in C(\CONT)_{\bool}$

\item
$C(\CONT)_{\bool\arrow\bool\arrow\bool}=\{\por\}$

\item 
$\por\ \true\ M\ \arrow^\delta \true$, 
$\por\ M\ \true \arrow^\delta \true$
\item
$\por\ \false\ \false \arrow^\delta \false$
\end{itemize}

\end{definition}

\begin{lemma}\label{lemma:cont}
$\Cont$ is a fully complete model of $\CONT$.
\end{lemma}

\begin{proof}
This is a corollary of Plotkin's proof of full abstraction 
of the Scott model of parallel PCF \cite{Plotkin77}. 
\end{proof}

\begin{lemma}\label{lemma:latcont}
$\Lat$ is a  model of $\CONT$.
\end{lemma}
\begin{proof}
The ground constants $\true$ and $\false$ and $\Omega$ are interpreted by 
$\ttt$, $\ff$ and $\perp$ respectively. Here there are  the 
interpretations of $\ifs$ and $\por$ in $\Lat$:

\vspace{0.3cm}

$
\Inte{\ifs}{}{\Lat} d\ e\ f=\left\{
\begin{array}{ll}
\perp &\mbox{ if } d=\perp\\
e&\mbox{ if } d=\ttt\\
f&\mbox{ if } d=\ff\\
\top &\mbox{ if } d=\top\\
\end{array}\right.
$ 


$\Inte{\por}{}{\Lat} d\ e=
\{\ttt\ |\ d\geq \ttt\mbox{ or } e\geq\ttt\} \vee 
\{\ff\ |\ d\geq \ff\mbox{ and } e\geq\ff\}\}
$ 

Otherwise stated, the interpretation of $\por$ in 
$\Lat$ is given by the following truth table:

$$\begin{array}{c|c|c|c|c|}
\por &\ \perp\ &\ \ff\ &\ \ttt\ &\ \top\ \\
\hline
\perp&\perp&\perp&\ttt&\ttt\\
\hline
\ff&\perp&\ff&\ttt&\top\\
\hline
\ttt&\ttt&\ttt&\ttt&\ttt\\
\hline
\top&\ttt&\top&\ttt&\top\\
\hline
\end{array}
$$


 
\vspace{0.3cm}

This interpretation validates the $\delta$-rules concerning
$\ifs$ and $\por$.

Notice that the interpretation of the parallel-or constant in $\cal L$
given above is the only one which is sound with respect to the 
$\delta$-rules of $\por$. Actually:

$\ttt=\Int{\true}^{\cal L}=\Int{\por\ \true\ x}^{\cal L}_{[x:=\top]}   =
\Int{\por}^{\cal L}\ttt\top
$

\end{proof}








\section{Extensional Collapse Situations}\label{main}

In this section, we prove the six extensional collapse situations 
exhibited in Figure \ref{ecs}.
An effective way of providing  an extensional collapse situation
$\coll {\cal M}{\cal N}{\rm  E}$, is to define an  
applied $\lambda$-calculus $\Lambda^{\cal N}$ such that 
$\cal N$ is a fully complete model of  $\Lambda^{\cal N}$,
and ${\cal M}$ is a model of  $\Lambda^{\cal N}$; then 
full completeness may be used to exhibit a suitable 
pre-logical surjection.
The key fact is expressed by the fundamental lemma of (pre-)logical
relations below: 

\begin{lemma}[\cite{HonsellS99}, Lemma 4.1] \label{lemma:fundamental}
If $\cal M$, $\cal N$ are  models of an applied $\lambda$-calculus $\Lambda$, 
and 
 $\rm E$ is  a pre-logical relation between  $\cal M$ and  
$\cal N$, such that
for all  constants $c:\sigma$ of $\Lambda$,
 $(\Inte{c}{}{\cal M},\Inte{c}{}{\cal N})\in{\rm E}_\sigma$, 
then, for all types $\sigma$ and for all $P\in\Lambda^0_\sigma$,
$(\Int{P}^{\cal M},\Int{P}^{\cal N})\in E_\sigma$.
\end{lemma}

\begin{corollary}\label{corollary:fundamental}
Let $\cal M$, $\cal N$ be type frames over $K$, and suppose that 
there exists an applied $\lambda$-calculus $\Lambda$
such that $\cal N$ is  fully complete for $\Lambda$.

Then there exist a pre-logical surjection $E$ such that $\coll {\cal M}{\cal N}{\rm E}$.





\end{corollary}

\begin{proof}

  Define $E_\sigma= \{(\Int{P}^{\cal M},\Int{P}^{\cal N})\ | \ P\in\Lambda_\sigma^0\}$.

As shown in \cite{HonsellS99}, Exemple 3.7, this is a pre logical
relation.

We have to prove that ${\rm E}_\sigma$ is a partial surjective function,
for all $\sigma$. For ground types, the statement holds
by full completeness of $\cal N$, under the hypothesis that 
different ground constants of $\Lambda$ are not identified in $\cal M$,
that we take for granted.

By using Lemma \ref{lemma:reduction} we 
are left with proving that ${\rm E}$ is surjective at higher types,
and this is guaranteed by the full completeness of $\cal N$.
\end{proof}

Given an applied $\lambda$-calculus $\Lambda$ and one of its models 
$\cal M$, let $Th^\Lambda({\cal M})=\{P=Q\ | \ P,Q\in \Lambda^0\mbox{ and }
 \Int{P}^{\cal M}=\Int{Q}^{\cal M}\}$ .

\begin{corollary}\label{corollary:theories}
Let $\cal M$ and $\cal N$ be models of an applied $\lambda$-calculus $\Lambda$,
and  ${\cal M}\arrow^E{\cal N}$ be an extensional collapse 
situation. If, for all constants $c:\sigma$ of $\Lambda$, 
$(\Int{c}^{\cal M},\Int{c}^{\cal N})\in E_\sigma$, then 
$Th^\Lambda({\cal M})\subseteq Th^\Lambda({\cal N})$.
\end{corollary}
\begin{proof}
Let us suppose that 
 $\Int{P}^{\cal M}=\Int{Q}^{\cal M}$, for two given closed 
$\Lambda$-terms $P$ and $Q$. By lemma \ref{lemma:fundamental}
we have that $(\Int{P}^{\cal M},\Int{P}^{\cal N}),
(\Int{Q}^{\cal M},\Int{Q}^{\cal N})\in E_\sigma$. Since $E_\sigma$ 
is a function, we conclude that  $\Int{P}^{\cal N}=\Int{Q}^{\cal N}$.
\end{proof}
Here is our main result: 
\begin{proposition}\label{main}

The following extensional collapse situations do hold:

$\coll \Cont \Set {}$,
$\coll \Err  \Set {}$,
$\coll \Cont \Err {}$,
$\coll \Err  \Cont {}$,
$\coll \Lat \Cont {}$,
$\coll \Lat \Err {}$.


\end{proposition}
\begin{proof}

\begin{itemize}
\item $\coll \Cont \Set {}$ and $\coll \Err  \Set {}$ follow from 
Corollary \ref{corollary:fundamental}, where the applied $\lambda$-calculus is $\SET$, and Lemmas 
\ref{lemma:set}, \ref{lemma:conterrset}.
\item  $\coll \Cont \Err {}$ and $\coll \Err  \Cont {}$
follow from Proposition \ref{prop:errcont}

\item $\coll \Lat \Cont {}$
follows from 
Corollary \ref{corollary:fundamental}, 
 where the the applied $\lambda$-calculus is $\CONT$,
and Lemmas
\ref{lemma:cont}, \ref{lemma:latcont}.
\item $\coll \Lat \Err {}$
follows from the two items above and Proposition  \ref{prop:comp}.
\end{itemize}
Note that, by composition of extensional collapse situations, 
we obtain $\coll \Lat \Set {}$. 
\end{proof}
By Corollary \ref{corollary:theories}, the proposition above 
yields in particular:
\begin{corollary}
\begin{itemize}
\item $\Th^{\Lambda^S}({\cal C})\subseteq \Th^{\Lambda^S}({\cal S})$
\item $\Th^{\Lambda^C}({\cal L})\subseteq \Th^{\Lambda^C}({\cal C})$
\end{itemize}
\end{corollary}
The examples presented in Section \ref{section:introduction} show that
the inclusions above are strict.
\section{Conclusion and further work}\label{section:conclusion}
We have provided a general definition of ``inclusion'' of models
of higher order, functional computations via the notion of 
extensional collapse situation, and shown this notion at work
on some simple models 
over $\bool$.






As suggested by the title, this work should  be considered as 
a first step 
settling the ground for the study of extensional collapse 
situations on more complicated models.
As a matter of fact, 
the definitions of section \ref{sec:type frames} have to be 
generalised a bit to encompass non extensional models. In particular, if
$\cal C$ is a Cartesian closed category 
  without enough points, the corresponding type 
frame ${\cal C}_\sigma$ is  defined by:
\begin{itemize}
\item ${\cal C}_\bool={\cal C}(1,C_\bool)$ for a suitable object $C_\bool$
\item ${\cal C}_{\sigma\arrow\tau}={\cal C}(1, 
{\cal C}_\sigma\Rightarrow {\cal C}_\tau )$
\end{itemize}

and, given $f\in {\cal C}_{\sigma\arrow\tau}$ and $x\in {\cal C}_{\sigma}$,
one defines $f(x)=ev\circ<\!\!f,x\!\!>$. 

Then, the fundamental lemma \ref{lemma:fundamental} holds and the 
construction of extensional collapse situations via applied 
$\lambda$-calculi goes through. This generalisation of logical 
relations for categories without enough points is described for instance 
in \cite{AmadioC98}, section 4.5.

 Let us introduce a direction for further work via 
an example, referring to the ones showed in the introduction. The terms:

{\tt $\lambda$x:bool. x}

{\tt $\lambda$x:bool. if x then x else x}

get the same interpretation in all the models considered in the 
present paper. They get different interpretation, for instance\footnote
{Actually, these terms get different interpretations in all 
models for which  {\em linearity} is an  extensional property.
This character of the model is often termed ``resource sensitivity''.
}, 
in the model $\Rel$ of set and relations \cite{Girard88}, 
endowed with the finite multiset exponential (see for instance
\cite{BucciarelliEM07} for a direct 
description of that model as a Cartesian
closed category).

Using Ehrhard's result on the extension collapse of 
$\Rel$ over $\Lat$ \cite{Ehrhard09}, we should obtain the 
following chain of extensional collapse situations:

$\Set<\Cont\equiv\Err<\Lat<\Rel$

An interesting sub-problem is the definition 
of an extension of (finitary) PCF w.r.t. which $\Lat$
is fully complete.

What comes after? More intensional models, like
game models, could collapse over $\Rel$.
Also, where is the place of models based on stable orderings, 
as the stable or strongly stable ones, in this hierarchy?

We already know some partial answer (for instance, the fact that
Berry and Curien's sequential algorithms collapse over the
model of strongly stable functions \cite{Ehrhard96a}).

Quite a number of different models of higher order, functional 
computation exist in literature. 

Having a better  overview  of the poset of extensional 
collapse situations relating these models may contribute to a better
understanding of the whole picture.
\bigskip
\subsubsection*{Acknowledgments. } 
Thanks to an anonymous referee of a previous version of this 
work, for suggesting that pre-logical relations are the 
proper framework to deal with extensional collapse situations.

\bibliographystyle{plain}
\bibliography{mybibDBLP}
\end{document}

%% file: upcf.pstex_t
\begin{picture}(0,0)%
\includegraphics{upcf.pstex}%
\end{picture}%
\setlength{\unitlength}{4144sp}%
\begingroup\makeatletter\ifx\SetFigFont\undefined%
\gdef\SetFigFont#1#2#3#4#5{%
  \reset@font\fontsize{#1}{#2pt}%
  \fontfamily{#3}\fontseries{#4}\fontshape{#5}%
  \selectfont}%
\fi\endgroup%
\begin{picture}(1515,1474)(5341,-4760)
\put(5626,-3301){\makebox(0,0)[b]{\smash{{\SetFigFont{12}{14.4}{\rmdefault}{\mddefault}{\updefault}{\color[rgb]{0,0,0} }%
}}}}
\put(6121,-3481){\makebox(0,0)[b]{\smash{{\SetFigFont{12}{14.4}{\rmdefault}{\mddefault}{\updefault}{\color[rgb]{0,0,0}$\Lat$}%
}}}}
\put(5356,-4021){\makebox(0,0)[b]{\smash{{\SetFigFont{12}{14.4}{\rmdefault}{\mddefault}{\updefault}{\color[rgb]{0,0,0}$\Err$}%
}}}}
\put(6841,-4021){\makebox(0,0)[b]{\smash{{\SetFigFont{12}{14.4}{\rmdefault}{\mddefault}{\updefault}{\color[rgb]{0,0,0}$\Cont$}%
}}}}
\put(6166,-4696){\makebox(0,0)[b]{\smash{{\SetFigFont{12}{14.4}{\rmdefault}{\mddefault}{\updefault}{\color[rgb]{0,0,0}$\Set$}%
}}}}
\end{picture}%

%% file: main.bbl
\begin{thebibliography}{10}

\bibitem{AbramskyJM00}
Samson Abramsky, Radha Jagadeesan, and Pasquale Malacaria.
\newblock Full abstraction for {PCF}.
\newblock {\em Inf. Comput.}, 163(2):409--470, 2000.

\bibitem{AmadioC98}
Roberto Amadio and Pierre-Louis Curien.
\newblock {\em Domains and lambda-calculi}, volume~46 of {\em Cambridge Tracts
  in Theoretical Computer Science}.
\newblock Cambridge University Press, 1998.

\bibitem{Bucciarelli97c}
Antonio Bucciarelli.
\newblock Bi-models: Relational versus domain-theoretic approaches.
\newblock {\em Fundam. Inform.}, 32(3-4):251--266, 1997.

\bibitem{BucciarelliEM07}
Antonio Bucciarelli, Thomas Ehrhard, and Giulio Manzonetto.
\newblock Not enough points is enough.
\newblock In Jacques Duparc and Thomas~A. Henzinger, editors, {\em CSL}, volume
  4646 of {\em Lecture Notes in Computer Science}, pages 298--312. Springer,
  2007.

\bibitem{BucciarelliS98}
Antonio Bucciarelli and Ivano Salvo.
\newblock Totality, definability and boolean ciruits.
\newblock In Kim~Guldstrand Larsen, Sven Skyum, and Glynn Winskel, editors,
  {\em ICALP}, volume 1443 of {\em Lecture Notes in Computer Science}, pages
  808--819. Springer, 1998.

\bibitem{Ehrhard96a}
Thomas Ehrhard.
\newblock Projecting sequential algorithms on strongly stable functions.
\newblock {\em Annals of Pure and Applied Logic}, 77:201--244, 1996.

\bibitem{Ehrhard09}
Thomas Ehrhard.
\newblock The {S}cott model of linear logic is the extensional collapse of its
  relational model.
\newblock Technical report, 2009.
\newblock HAL hal-00369831, to appear in Theoretical Computer Science.

\bibitem{Girard88}
Jean-Yves Girard.
\newblock Normal functors, power series and lambda-calculus.
\newblock {\em Annals of Pure and Applied Logic}, 37(2):129--177, 1988.

\bibitem{HonsellS99}
Furio Honsell and Donald Sannella.
\newblock Pre-logical relations.
\newblock In J{\"o}rg Flum and Mario Rodr\'{\i}guez-Artalejo, editors, {\em
  CSL}, volume 1683 of {\em Lecture Notes in Computer Science}, pages 546--561.
  Springer, 1999.

\bibitem{HylandO00}
J.~M.~E. Hyland and C.-H.~Luke Ong.
\newblock On {F}ull {A}bstraction for {PCF}: {I}, {II}, and {III}.
\newblock {\em Inf. Comput.}, 163(2):285--408, 2000.

\bibitem{OHearnR95}
Peter~W. O'Hearn and Jon~G. Riecke.
\newblock Kripke logical relations and {PCF}.
\newblock {\em Information and Computation}, 120(1):107--116, 1995.

\bibitem{Plotkin80}
Gordon Plotkin.
\newblock Lambda definability in the full type hierarchy.
\newblock In J.~P. Seldin and J.~R. Hindley, editors, {\em To H.B.Curry :
  essays on combinatory logic, lambda calculus, and formalism}, pages 363--373.
  Academic Press, 1980.

\bibitem{Plotkin77}
Gordon~D. Plotkin.
\newblock {LCF} considered as a programming language.
\newblock {\em Theor. Comput. Sci.}, 5(3):225--255, 1977.

\bibitem{Sieber92}
Kurt Sieber.
\newblock Reasoning about sequential functions via logical relations.
\newblock In M.~Fourman, P.~Johnstone, and A.Pitts, editors, {\em Proceedings
  of the LMS Symposium on Applications of Categories in Computer Science},
  volume 177 of {\em {LMS} Lecture Notes Series}. Cambridge University Press,
  1992.

\end{thebibliography}
